\documentclass{article}
\usepackage{amsmath}
\usepackage{amssymb}
\usepackage{amsthm}
\usepackage{bbm}
\usepackage{mathtools}
\usepackage{tikz}
\usetikzlibrary{trees}

\usepackage[style=alphabetic,backend=bibtex,url=false,doi=false,isbn=false]{biblatex}
\usepackage{hyperref}  
\usepackage{nameref}
\usepackage{cleveref}
\usepackage{algorithm2e}
\usepackage{seqsplit}
\usepackage{braket}
\usepackage{romannum}
\usepackage{float}
\usepackage{booktabs,tabularx}

\usepackage{graphicx}
\usepackage[T1]{fontenc}

\RestyleAlgo{ruled}
\SetKwComment{Comment}{/* }{ */}

\newtheorem{definition}{Definition}
\newtheorem{theorem}{Theorem}
\newtheorem{lemma}{Lemma}
\newtheorem*{theorem*}{Theorem}

\title{Quantum Algorithms for the Minimum Steiner Tree problem with application to Binary Near-Perfect Phylogenies}
\author{Lingfa Meng$^1$, David Salvador Novo$^2$, Albert H. Werner$^{*3}$, Samir Bhatt$^{*4}$}
\date{$^1$Department of Public Health, University of Copenhagen\\[2ex] $^2$Department of Mathematics, University of Copenhagen\\[2ex] $^3$QMATH, University of Copenhagen\\[2ex] $^4$Department of Public Health, University of Copenhagen\\[2ex]$^*$ joint senior authors, \\[2ex]%
\today}

\begin{document}

\pagenumbering{arabic}

\maketitle
\section{Introduction}

Despite decades of efforts into quantum algorithms, there has been a lack of a concrete algorithm with an end-to-end analysis for a problem of practical interest. We first present a quantum algorithm in bioinformatics for solving the Binary Near-Perfect Phylogeny Problem (BNPP) with a complexity bound of $O(8.926^q + 8^qnm^2)$, where n is the number of input taxa and m is the sequence length for each taxon with each character in the sequence being a binary bit using the QRAM model. We give another polynomial space exact algorithm for the Minimum Steiner Tree (MST) problem with complexity \( O^*(e^{(1+{g(k,l)})k})\) in the circuit model.

Section 1 is the introduction. Section 2 defines BNPP and reviews the current best classical algorithms for exact BNPP. Section 3 presents a quantum algorithm for BNPP under the QRAM model. Section 4 presents a polynomial space quantum algorithm for MST under the circuit model. Section 5 is the conclusion and future directions.

\section{Binary Near-Perfect Phylogeny}

Phylogenetic trees serve as foundational structures within evolutionary biology, enabling the representation of shared histories arising from the evolutionary process. Modern computational phylogenetics seeks to infer a most probable binary (rooted or unrooted) tree from genetic data comprised of a set of characters (protein, codon or nucleotide). A popular paradigm for inference relies on an optimality criterion where the most likely phylogenetic tree minimises the total number of character-state changes\cite{Fitch1967-ms,Fitch1971-tz} - termed maximum parsimony or the method of minimum evolution\cite{Edwards1996-lb}. While initially a popular approach for phylogenetic analysis, parsimony methods have largely been superseded by likelihood based approaches which have shown to have a better performance across a broad range of evolutionary scenarios\cite{Steel2000-rp}. However, there are still cases where the non-parametric approach of maximum parsimony is preferable or competitive\cite{Kolaczkowski2004-bw}.

A facet of parsimony that has been studied in depth is where the genetic sequence is assumed to be infinite sites in length such that no site mutates more than once across the entire tree - as known as an absence of homoplasy. A tree which adheres to no homoplasy is known as a perfect phylogeny. Clearly, the absence of homoplasy assumption is both highly restrictive and implausible even for very long sequences with billions of characters. Perfect phylogenies are of practical interest in solving the haplotyping problem\cite{Gusfield2002-im} where the challenge is to distinguish between haplotypes given genetic data, in trait association\cite{Mailund2006-ua}.

The computational bottleneck is identifying the perfect phylogeny (PP) or near-perfect phylogeny (NPP) - where all characters evolve at sites over the tree with no homoplasy or minimum homoplasy. The PP framework is used for tree reconstruction in cancer cell sequencing \cite{Jahn_2016, Malikic:2019aa} and haplotyping  \cite{bafna_haplotyping_2003}. NPP extends PP for input data not allowing a PP. A recent study \cite{10.1093/sysbio/syab069} has identified viral epidemics to be a good setting for NPP application where the false positive rate of obtaining splits in the reconstructed tree \(\widehat{T}\) compared to the true tree \(T\) is no higher than that of Maximum Likelihood.

\subsection{PP and BNPP}
The definition for a phylogeny is given as follows.

\begin{definition}[Phylogeny \cite{sridhar_algorithms_2007} ]\label{def:phylo}

A tree $T=(V, E)$ such that

1. The hamming distance between any two connected nodes is 1.

2. Every input taxon appears in the tree and all nodes with degree no bigger than 2 is an input taxon.

\end{definition}

% what is a shallow time phylogeny? a precise definition will be good
This definition for a phylogeny \ref{def:phylo} differs from another common definition, where the set of leaf nodes are identified with the set of input taxa. The flexibility for any taxa to be an internal node of degree-2 accommodates for shallow time phylogenies, where BNPP shows a comparable accuracy to likelihood-based methods \cite{10.1093/sysbio/syab069}. We note that many likelihood methods restrict input taxon to be non-internal nodes. We then define the perfect phylogeny problem.

\begin{definition}[Perfect Phylogeny]\label{prob:pp}
    
Given an n by m matrix M with integers from the set 
\begin{equation*}
K = \{ { 1,2,...,k} \}
\end{equation*}
 where each row of M is referred to as a taxon, each column c of M is referred to as a character, and each value in a cell in column c is referred to as a state of character c, reconstruct an undirected tree $T$ where each leaf is labelled by a distinct taxon in M and each internal node by a vector in $K^m$, such that all nodes identified by any state $i$ of a character $c$ form a connected subtree.

\end{definition}

Note that every leaf in $T$ will be labelled by a row vector in $M$, which is a direct corollary of the connected subtree constraint. This problem \ref{prob:pp} has been proven to be NP-complete \cite{goos_two_1992}, and fixed-parameter tractable (FPT). For $k=2$, we obtain the special case of a Binary Perfect Phylogeny (BPP). We next define BNPP which is a more general case of PP. 

 Previous work identifies polynomial bounds in $n,m$ given a fixed $k$ for PP. Especially for the $k=3$ case \cite{DRESS19923} have given a polynomial bound. The definition for the BNPP problem is taken from \cite{fernandez-baca_polynomial-time_2003}.

\begin{definition}[Binary Near-Perfect Phylogeny]
\label{prob:binary_near_perfect_phylogeny}
Let $l$ be a labelling function from a taxon (vertex) to a boolean string. L denotes the set of all possible taxa with $m$ characters.
$$\forall \ v \in L, \ \exists \ l(v) \in \left\{ 0,1\right\}^ {m}$$
We say that there exists a q-near-perfect phylogeny when there exists a phylogenetic tree $T$ that attains q for the following penalty function.
$$\textrm{penalty} \ (T) = \sum_{(u,v)\in E(T)} d(l(u), l(v)) - m$$
where $d$ is the hamming distance.
\end{definition}

Intuitively, q could be the number of times that the system displays a homoplasy, which is the independent evolutions of two different taxa on the same character leading to the same state. Purely for the convenience of definition, we shall assume that at each time step only one mutation happens and this mutation exhibits a homoplasy if the result of the mutation leads to a character state that is already present from previous mutations. This penalty function also accounts for the number of back mutations in a tree, however a solution with a back mutation whose removal only reduces the penalty by 1 is never accepted as it does not violate any four-gamete conditions.

The penalty q is the critical term governing the complexity of BNPP. q is linearly related to the parameter k - number of vertices to be connected in MST.

\subsection{Classical Algorithm for BNPP}

We review the second algorithm given in \cite{sridhar_algorithms_2007}, which was the most efficient fixed-parameter tractable algorithm for BNPP, with a complexity of $O(21 ^ q + 8 ^ q n m ^ 2)$, where q is the value of the penalty function for the output tree, n is the number of input taxon and m is the number of sites. We give the definition of a conflict graph whose cardinality dominates the size of MST instances.

\begin{definition}[Conflict Graph]
\label{def:conflict_graph}
For \[I \in M_{n \times m}, \ \forall \ i\  j,\ I_{ij} \in \{ 0,1\}\]
let the set of columns be the set of graph vertices and two vertices be connected if and only their corresponding columns satisfy the four-gamete conditions. The four-gamete condition states that two colomns exhibit all four possible pairwise combinations.
\[\{ 0,0\}, \{ 1,0\}, \{ 0,1\}, \{ 1,1\}\]
The conflict graph is the union of the non-isolated components of the above graph.
\end{definition}
We reproduce the following pseudocode \ref{alg:classical_bnpp} from \cite{sridhar_algorithms_2007}, which has a complexity of $O(21^q + 8^q nm^2)$ in the derandomized case.

\begin{algorithm}
\caption{BuildNPP}\label{alg:classical_bnpp}

\KwData{$I \in M_{n \times m}, \ \forall \ i\  j,\ I_{ij} \in \{ 0,1\}$}
\KwResult{$E \cup ( \cup_i T_i)$}
$L := \{ I \}$\;
$E:=\emptyset $\;
1) \While{$| \cup_{M_i \in L} N(M_i)  | > q$}{
	a. guess vertex $v$ from $\cup _ {M_i \in L } N(M_i)$, let $v \in N(M_j)$, with probability at least $4^{-q}$ that all guesses succeed\;
	b. let $M0 := M_j (c(v), 0)$ and $M1 := M_j (c(v), 1)$\;
	c. guess taxa $r$ and $p$, with probability at least $2^{-q}$ that all guesses succeed\;
	d. add $r$ to $M1$, $p$ to $M0$ and $(r, p)$ to $E$\;
	e. remove $M_j$ from $L$, add $M0$ and $M1$ to $L$\;
}
2) for each $M_i \in L$ compute an optimum phylogeny $T_i$\;

\end{algorithm}

The algorithm \ref{alg:classical_bnpp} is divided into 2 parts:

1) Divide-and-conquer: Row-split the input taxa matrix $M$ into a set of matrices $\{M_i\}$ to reduce the number of non-isolated vertices in the conflict graph for each $M_i$.

2) Dynamic programming: Obtain a MST ${T_i}$ using characters in the connected component of the conflict graph of $\{M_i\}$ when 1) bottoms out. The union of all $T_i$, with the rest of the perfect characters added using Gusfield’s linear time algorithm outputs the final phylogeny.

\subsection{Asymptotic growth of the number of homoplasms with the number species and number sites $q(n, m)$}

To understand the complexity of BNPP, we study the dependence of distance $q$ on the input taxon set of size $n$ and character set of size $m$, i.e. $q(n, m)$. We show below that the expected number of homoplasm events scales linearly in $n$ and $m$, assuming a pure birth branching process for the evolution process. We present the scaling for a toy model in \ref{sec:q_asymptotic_growth} in the appendices.

For an observed set $S$ of taxa of size $n$, assuming a pure birth branching process leading to this set of taxa, the number of nodes in the phylogenetic tree leading to $S$ is of order $O(n)$. Given that the expected number of homoplasm events scale linearly with the number of nodes in the phylogenetic tree as shown in \ref{sec:q_asymptotic_growth}'s calculation.

The number of characters that has experienced one mutation  is linear in the number of characters and time since we assume that the mutation of each site follows an independent poisson distribution. The number of homoplasm is exponential in time as the tree area grows exponentially in time.

\section{Quantum algorithm for Binary Near Perfect Phylogeny}

There has been recent interest in combining techniques from classical algorithm design such as dynamic-programming with quantum search. \cite{doi:10.1137/1.9781611975482.107} proved the first result in combining Grover's search with dynamic programming and \cite{10.1007/978-3-030-58150-3_19} applied the same method to another dynamic programming problem called the Minimum Steiner Tree (MST) problem. We use the MST quantum subroutine of \cite{10.1007/978-3-030-58150-3_19} to replace the classical DP subroutine. We give a plugin lemma \ref{lem:mst_plugin_bnpp} which we use to compare the different results.

\begin{lemma}[MST to BNPP]\label{lem:mst_plugin_bnpp}
    Given a MST algorithm running in time \(O^*(C^k)\), there exists a BNPP algorithm running in time \(O^*((2C^2\frac{2C^2+1}{2C^2-1})^q+8^qnm^2)\). k is the number of vertices to be connected for MST. q is the near-perfect number for the BNPP solution. n is number of taxa and m is the number of characters.
\end{lemma}

The proof to this lemma is given in appendix \ref{sec:plug_in}. We will plug in the following result from \cite{10.1007/978-3-030-58150-3_19}.

\begin{theorem}[Theorem 1 of \cite{10.1007/978-3-030-58150-3_19}]\label{the:q_mst}
There exists a quantum algorithm that solves with high probability the Minimum Steiner Tree problem in time $O^*(1.812^k)$, where k denotes the size of the terminal set.
\end{theorem}
We obtain the following result.
\begin{theorem}[Quantum Algorithm for BNPP]\label{the:q_bnpp}
There exists a quantum algorithm that solves with high probability BNPP in time $O^*(8.926^q + 8^q nm^2)$, where n is the number of taxa, m is the number of characters and q is the distance from a perfect phylogeny. The algorithm uses the QRAM model.
\end{theorem}

\subsection{Grover's algorithm for Dynamic Programming}
We explain the relationship between MST and BNPP, then state the complexity advantage of the quantum algorithm relative to the classical s. We define the problem of a minimum Steiner Tree (MST) below which is equivalent to BNPP \ref{prob:binary_near_perfect_phylogeny} barring property 2 in \ref{def:phylo}.

\begin{definition}[Minimum Steiner Tree]\label{prob:minimum_steiner_tree}
Given a graph $G = (V, E)$, a function $length: 2^E \rightarrow \mathbb{R}$ and a subset of vertices $V' \subseteq V$, find a subgraph $G' = (V'', E'),\ V' \subseteq V'' \subseteq V$ where $G'$ satisfies the following property.

$$W_G(V') = \arg \min_{G'=(V', E')} length(E')$$

$W_G(V')$ is the weight of a minimum Steiner tree connecting $V'$ in $G$.
\end{definition}

More specifically, we'd like to consider the problem of finding the MST on a graph embedded in a hypercube.

\begin{definition}[Minimum Steiner Tree on a hypercube]
\label{prob:mst_on_hypercube}
Let $l$ be a labelling function from a lattice point on a hypercube of dimension m to a boolean string.
$$\forall \ v \in L, \ \exists \ l(v) \in \left\{ 0,1\right\}^ {m}$$
G denotes the graph structure induced by the m-dimensional hypercube. For an input set $I$ of lattice points from the m-dimensional hypercube, identify $W_G(I)$, the weight of the minimum Steiner Tree with input vertex set $I$ on graph $G$.
\end{definition}

Definitions \ref{prob:binary_near_perfect_phylogeny} and \ref{prob:mst_on_hypercube} show that MST is equivalent to BNPP. MST obtains a speedup using Grover's algorithm and dynamic programming \cite{10.1007/978-3-030-58150-3_19} with \(C=1.812\). This brings the overall complexity from $O(21 ^ q + 8 ^ q n m ^ 2)$ in \cite{sridhar_algorithms_2007} to $O( 8.926^ q + 8 ^ q n m ^ 2)$. The best classical case using MST algorithm from \cite{10.1145/1250790.1250801} with \(C=2\) and the framework of \cite{sridhar_algorithms_2007} yields $O(10.286^q+8^qnm^2)$. We naturally ask as a next step how this algorithm can be implemented and if the advantage can be preserved.%The next step is to try to recover the full Grover speedup for the Quantum MST or understand why it cannot be done.

\subsection{Implementation of the quantum dynamic programming}

To implement the above mentioned quantum algorithm for MST using dynamic programming, the two pressing issues are oracle efficiency and classical data encoding efficiency. 
Adopting the same multi-level recursive search with measurement between layers yields an encoding complexity that eliminates the advantage. 
The encoding cost alone will be $\Theta^*(2^k)$ since the cost of each each node in the Dreyfus-Wagner (DW) recursion will need to be encoded. DW recursion is a classical recursion solving the MST problem through dynamic programming and will serve as the basis of the quantum search for MST. The form of DW recursion as applied in \cite{10.1007/978-3-030-58150-3_19} is presented below. 

\begin{equation*}\label{def:dw_encoding_cost}
W_G(K)=
\min _ {\substack{ K_1 \subseteq  K \\ |K_1| =  (\alpha \pm \eta) }}
\min _ {\substack{ A \subseteq  V \\ |A| \leq  \lceil \log 1/\epsilon \rceil }}
W_{G}(K_1 \cup A) + W_{G/A}(K_2 \cup \{v_A\}) 
\end{equation*}

One solution to the two challenges is to create the full solution subspace, where each solution is indexed by the branching selections on all levels, and a single Grover's search is run on this superposition of solutions. Each candidate solution is indexed by the choices in the sets of $\{K_1\}, \{A\}$ across the different layers.

We describe our approach which is not necessarily the optimal approach. A solution is indexed by a combination of different branches on the 3 levels of DW recursion. A weight register accompanies each branching or terminal node in the 3 levels of DW recursion.

 \begin{figure}[h!]
 \centering
  \includegraphics[width=1\textwidth]{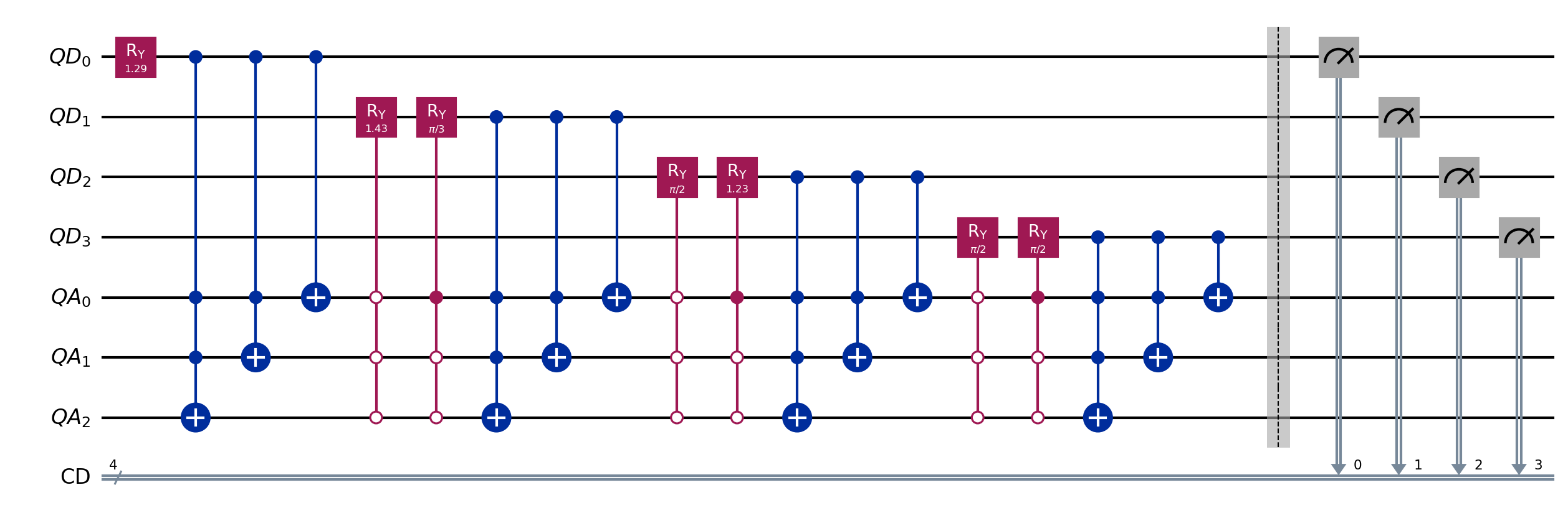}
 \caption{Dicke state preparation circuit over 4 qubits}
 \end{figure}
 
%  \begin{figure}[h!]
% \centering
%  \includegraphics[width=0.6\textwidth]#{QuantumSearchCircuit}
% \caption{Quantum Search circuit}
% \end{figure}
 
%  \begin{figure}[h!]
% \centering
%  \includegraphics[width=0.8\textwidth]{MinimumFinding}
% \caption{Quantum Minimum finding circuit}
% \end{figure}

\subsection{State preparation for 3-level DW recursion}
We construct the index qubits top down. The search space is spanned by a superposition of Dicke states of the right hamming weights, which is governed by the $\eta$ parameter, where $\log \frac{1}{\eta}$ specifies the size of the subgraphs $A$ to be minimized over for. We use Dicke states to achieve efficient encoding of the solution space as opposed to using the full space generated by Hadamards. Dicke states can be efficiently and deterministically prepared using the procedure from \cite{10.1007/978-3-030-25027-0_9} with a gate complexity of $O(nk)$, where $n$ is the total number of qubits and k is the occupation number. A more generalized Dicke state preparation for a superposition of arbitrary Dicke states can be efficiently achieved using \cite{icissp23}.

% Need more explaination regarding this, is it trying to say that hadamard encoding makes the encoding complexity squred relative to the Dicke states based encoding?
% The binary encoding has an extra multiplicative factor of $\sqrt{2}$ in time complexity analysis.
\begin{figure}[H]
\centering

\begin{tikzpicture}[
  ->,
  xscale=1.5,
  grow=down,
  level 1/.style={sibling distance=40mm},
  level 2/.style={sibling distance=20mm},
  level 3/.style={sibling distance=15mm},
  level distance=15mm,
  every node/.style = {}
]
\node (r) {\( \{\ket{K_1^1 , A_1^1}\}_{K_1^1 , A_1^1} \)}
  child { node (21) {\( \{ \ket{K_2^1 , A_2^1}\}_{K_2^1 , A_2^1} \)}
    child { node (31) {\( \{\ket{K_3^1 , A_3^1}\}_{K_3^1 , A_3^1} \)} 
      edge from parent node[left, draw=none] {\( \substack{K_3^1 \subseteq K_2^1 \cup A_2^1 \\ A_3^1 \subseteq V} \)}
    }
    child { node (32) {\( \{\ket{K_3^3 , A_3^2}\}_{K_3^3 , A_3^2} \)} 
      edge from parent node[right, draw=none] {\( \substack{K_3^3 \subseteq K_2^1 \cup \{v_{A_2^1}\} \\ A_3^2 \subseteq V} \)}
    }
    edge from parent node[left, draw=none] {\( \substack{K_2^1 \subseteq K_1^1 \cup A_1^1 \\ A_2^1 \subseteq V} \)}
  }
  child { node (22) {\( \{\ket{K_2^3 , A_2^2}\}_{K_2^3 , A_2^2} \)}
    child { node (33) {\( \{\ket{K_3^5 , A_3^3}\}_{K_3^5 , A_3^3}  \)} 
      edge from parent node[left, draw=none] {\( \substack{K_3^5 \subseteq K_2^3 \cup A_2^2 \\ A_3^3 \subseteq V} \)}
    }
    child { node (34) {\( \{\ket{K_3^7 , A_3^4}\}_{K_3^7 , A_3^4} \)} 
      edge from parent node[right, draw=none] {\( \substack{K_3^7 \subseteq K_2^3 \cup \{v_{A_2^1}\} \\ A_3^4 \subseteq V} \)}
    }
    edge from parent node[right, draw=none] {\( \substack{K_2^3 \subseteq K_1^1 \cup A_1^1 \\ A_2^2 \subseteq V} \)}
  };
\end{tikzpicture}
\caption{Constructing solution space for graph index qubits with subset relations shown on the edges. Each node represents a superposition of a set of sub-solution states. The arrows are directed top-down with set relations constraining the lower level sets by the higher level sets.}
\label{fig:tree-mst-graph-index}
\end{figure}
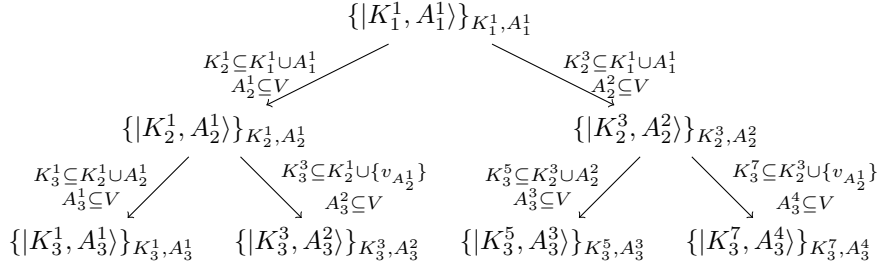

We construct the weight qubits bottom up. A weight register $W_3^b$ holds the weights of the classically computed trees on the bottom level. A first adder is then used to add up two weights with matching $A_3^b$ and $K_3^b$ according to \ref{sec:dw_recursion} to form a weight on the second level. A second adder is used to add up two weights with matching $A_2^b$ and $K_2^b$ to form a weight on the first level.

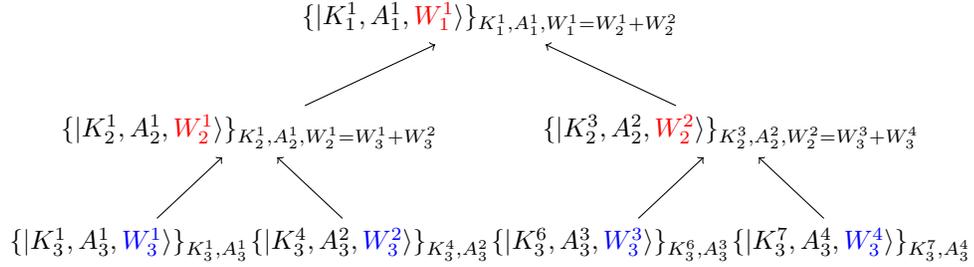
\begin{figure}[H]
\centering
\begin{tikzpicture}[
<-,
xscale=1.6,
  grow=down,
  level 1/.style={sibling distance=40mm},
  level 2/.style={sibling distance=20mm},
  level 3/.style={sibling distance=15mm},
  level distance=15mm,
  every node/.style = {}
]
\node (11) {\( \{\ket{K_1^1 , A_1^1, \textcolor{red}{W_1^1} }\}_{K_1^1 , A_1^1, W_1^1 = W_2^1 + W_2^2} \)}
  child { node (21) {\( \{\ket{K_2^1 , A_2^1, \textcolor{red}{W_2^1}}\}_{K_2^1 , A_2^1, W_2^1 = W_3^1 + W_3^2} \)}
    child { node (31) {\( \{\ket{K_3^1 , A_3^1, \textcolor{blue}{W_3^1}}\}_{K_3^1 , A_3^1} \)}
    }
    child { node (32) {\( \{\ket{K_3^4 , A_3^2, \textcolor{blue}{W_3^2}}\}_{K_3^4 , A_3^2} \)}
    }
  }
  child { node (22) {\( \{\ket{K_2^3 , A_2^2, \textcolor{red}{W_2^2}}\}_{K_2^3 , A_2^2,W_2^2 = W_3^3 + W_3^4} \)}
    child { node (33) {\( \{\ket{K_3^6 , A_3^3, \textcolor{blue}{W_3^3}}\}_{K_3^6 , A_3^3} \)}
    }
    child { node (34) {\( \{\ket{K_3^7 , A_3^4, \textcolor{blue}{W_3^4}}\}_{K_3^7 , A_3^4} \)}
    }
  };

\end{tikzpicture}
\caption{Constructing weight qubits for solution space.}
\label{fig:tree-mst-graph-index-with-weight}
\end{figure}

\begin{figure}[H]
\centering
\begin{tikzpicture}[
->,
xscale=1.5,
  grow=down,
  level 1/.style={sibling distance=40mm},
  level 2/.style={sibling distance=20mm},
  level 3/.style={sibling distance=15mm},
  level distance=15mm,
  every node/.style = {}
]
\node (11) {\( \frac{1}{\sqrt{d}}\sum_{K_1^1 , A_1^1} \ket{K_1^1 , A_1^1, \textcolor{red}{W_1^1} }_{W_1^1 = W_2^1 + W_2^2} \)}
  child { node (21) {\(\bigotimes \sum_{K_2^1 , A_2^1} \ket{K_2^1 , A_2^1, \textcolor{red}{W_2^1}}_{ W_2^1 = W_3^1 + W_3^2} \)}
    child { node (31) {\( \bigotimes \sum_{K_3^1 , A_3^1} \ket{K_3^1 , A_3^1, \textcolor{blue}{W_3^1}} \)}
    }
    child { node (32) {\( \ket{K_3^4 , A_3^2, \textcolor{blue}{W_3^2}} \)}
    }
  }
  child { node (22) {\( \ket{K_2^3 , A_2^2, \textcolor{red}{W_2^2}}_{W_2^2 = W_3^3 + W_3^4} \)}
    child { node (33) {\( \bigotimes \sum_{K_3^6 , A_3^3} \ket{K_3^6 , A_3^3, \textcolor{blue}{W_3^3}} \)}
    }
    child { node (34) {\( \ket{K_3^7 , A_3^4, \textcolor{blue}{W_3^4}} \)}
    }
  };
%\draw[red, thick] (21) -- (22);

\end{tikzpicture}
\caption{The final quantum state ready for quantum minimum finding. The quantum state shall be read as if it were a 3 line equation top to down, left to right.}
\label{fig:tree-mst-graph-index-with-weight-state}
\end{figure}

The above procedure shows the existence of an efficient \(U_{idx}\) mapping \(| 0^n \rangle \) to a uniform superposition of the solution subspace indices, resting on the existence of efficient Dicke state preparation \cite{icissp23}. This implies an efficient projector \(\Pi_{sub}  = U_{idx} (|0^n\rangle \langle0^n| \otimes I )U_{idx}^{\dagger} \) to the solution subspace and an efficient reflector \(R_{sub} = 2\Pi_{sub}-I\) in that subspace. These ensure that we can efficiently implement a Grover oracle and carry out the Grover as part of the quantum minimum finding process to find the minimum \(W_1^1\).

\subsection{Encoding complexity of classical data}

The classical sequence data grows as $O(nm)$ where $n$ is the number of input sequences and $m$ is the sequence length. The number of Steiner weights to be encoded scales as $O(2^{(H(\frac{1-\beta}{4})+\frac{1-\beta}{4})}k)$ \cite{10.1007/978-3-030-58150-3_19}. The classical data needed to be encoded into a quantum state are the classical weights precomputed from stage 1 of the quantum algorithm for MST with their corresponding graph indices. The state used for computation is shown below, consisting of $14|V| + 7\lceil \log_2|W| \rceil$ qubits. We have not counted intermediate auxiliary qubits for encoding. They might be of importance when implementing the reflection for this relevant subspace though the auxiliary qubit count should be linear in the number of graph indexing qubits. The factor 14 comes from storing the subgraph information for the 2 subgraphs on each node across 3 levels. More generally, for a l-level DW recursion we need \(2(2^l-1)|V| + (2^l-1)\lceil \log_2|W| \rceil\) qubits excluding intermediate ancila.

The prepared state to be processed for amplitude amplification is shown below. All of the constraints regarding set relations cardinalities are obeyed according to \ref{sec:dw_recursion}. We have put each pair of sub-Steiner solutions on the same line. Blue indicates that the weights are classically pre-computed and red indicates that the weights are computed using a quantum adder.

\begin{equation}\label{def:search_space}
\begin{split}
\ket{\phi }=  \frac{1}{\sqrt{d}} \sum  \limits_{K_1^1 , A_1^1}  \ket{K_1^1 , A_1^1, \textcolor{red}{W_1^1} }  \\
\frac{1}{\sqrt{d}} \bigotimes   \sum \limits_{K_2^1 , A_2^1}  \ket{K_2^1 , A_2^1, \textcolor{red}{W_2^1}} \ket{K_2^3 , A_2^2, \textcolor{red}{W_2^2}} \\
\frac{1}{\sqrt{d}} \bigotimes   \sum \limits_{K_3^1 , A_3^1} \ket{K_3^1 , A_3^1, \textcolor{blue}{W_3^1}} \ket{K_3^4 , A_3^2, \textcolor{blue}{W_3^2}}\\
\frac{1}{\sqrt{d}} \bigotimes   \sum \limits_{K_3^6 , A_3^3} \ket{K_3^6 , A_3^3, \textcolor{blue}{W_3^3}}\ket{K_3^7 , A_3^4, \textcolor{blue}{W_3^4}}
\end{split}
\end{equation}

We repeat the encoding four times for the classical weight data which is computed in stage 1 in \cite{10.1007/978-3-030-58150-3_19} for a subset of reduced Steiner instances and the classical computation grows at worst asymptotically as the query complexity of the quantum part of the algorithm. This implies that the quantum speedup is preserved concerning data encoding using QRAM. We notice that there could be more efficient encoding schemes by encoding level 2 states in $\lceil -\log \frac{1}{\epsilon} \rceil$ qubits, since the relative position for level 2 is encoded in level 1. However, this recursive indexing is impossible for level 3 since the encoding circuits need to identify relative position in $G$ while being ignorant of level 1 and level 2 branches. Conditioning using full level 1 and level 2 branches' information in the recursive indexing case requires an encoding circuit whose depth scales linearly in the total number of branches for this 3-level search tree. We present algorithm \ref{alg:q_build_search_space} which builds a uniform superposition of all index states from the solution subspace.

\begin{algorithm}[h]
\caption{BuildSearchSpace}\label{alg:q_build_search_space}

\KwData{classically pre-computated weight-index pairs}
\KwResult{ \(\ket{\phi}\) \ref{def:search_space}}
a) Encode \(K\) in the auxiliary qubits;

b) Using Dicke state preparation to span level 1 subgraph indices, wipe auxiliary qubits;

c) Compute set operations as specified in \ref{sec:dw_recursion} to obtain 
the relevant subsets of nodes from which level 2 \(K_x^y \) shall be generated from and store in auxiliary qubits;

d) Using Dicke state preparation to span level 2 subgraph indices;

e) Compute set operations as specified in \ref{sec:dw_recursion} to obtain 
the relevant subsets of nodes from which level 3 \(K_x^y \) shall be generated from and store in auxiliary qubits;

f) Using Dicke state preparation to span level 3 subgraph indices;

g) Encode pre-computed weights to the 4 branches on level 3 with corresponding indices;

h) Use quantum adder to compute the level 2 and 1 weights;
\end{algorithm}

\section{Polynomial space exact quantum algorithm for the Minimum Steiner Tree problem}
The following tables compare the time complexities between classical algorithms with QRAM-free quantum algorithms for different types of MST and space resource bounds.
\begin{table}[H]
\centering
\caption{Weight type vs.\ quantumness for Steiner problem algorithms}
\label{tab:steiner-poly-space}
\renewcommand{\arraystretch}{1.2}
\begin{tabular}{lcc}
\toprule
& $\mathrm{poly}$\text{-space-classical} & $\mathrm{poly}$\text{-space-quantum}\\
\midrule
\(w : E(G) \;\to\; \mathbb{N}\) & \(O^*(2^k)\) \cite{10.1145/1250790.1250801} & \(O^*(2^k)\)% e.g., Alg C
\\
\(w : E(G) \;\to\; \mathbb{R}_{>0}\) & \(O^*(6.751^k)\) \cite{doi:10.1137/17M1140030} & \(O^*(2^k)\)% e.g., Alg E
\\
\bottomrule
\end{tabular}
\end{table}

\begin{table}[H]
\centering
\caption{Weight type vs.\ quantumness for Steiner problem algorithms}
\label{tab:steiner-exp-space}
\renewcommand{\arraystretch}{1.2}
\begin{tabular}{lcc}
\toprule
& $\mathrm{exp}$\text{-space-classical} & $\mathrm{exp}$\text{-space-quantum}\\
\midrule
\(w : E(G) \;\to\; \mathbb{N}\) & \(O^*(2^k)\) \cite{10.1145/1250790.1250801} & \(O^*(2^k)\)% e.g., Alg C
\\
\(w : E(G) \;\to\; \mathbb{R}_{>0}\) & \(O^*(2.684^k)\) \cite{fuchs2007speeding} & \(O^*(2^k)\)% e.g., Alg E
\\
\bottomrule
\end{tabular}
\end{table}

We prove the following theorems regarding polynomial space quantum algorithms for MST in the circuit model. We use an upper bound for the binomial coefficients from Chapter 10, Lemma 7, p309 of \cite{bok:MW} and a polynomial space exact algorithm for the Minimum Steiner tree problem of \cite{doi:10.1137/17M1140030}. Theorem \ref{the:q_spacetime_tradeoff_mst} gives a polynomial space exact quantum algorithm for solving the minimum Steiner tree on a hybrid computer different from a recursive Grover search. We note that classically the best polynomial space exact algorithm for MST with general weights is given by \cite{doi:10.1137/17M1140030} with an exponential scaling of \(O^*(4^{(1+\epsilon)k}n^{f_1(\epsilon)})\) while the best exponential space algorithm for MST with general weights is given by \cite{10.1007/11672142_46} with an exponential scaling of \( O^*((2+\delta)^k n^{f_2({\delta}^{-1})})\). The subscripts in \(f_1, f_2\) denote the corresponding scaling functions given in each paper. In table \ref{tab:steiner-poly-space} and \ref{tab:steiner-exp-space} the classical counterparts we picked are the best FPT algorithms available. Theorem \ref{the:polyspace_q_mst} can be obtained by letting \( l \) tend to the maximal value of \(\lfloor \log k \rfloor\).

\begin{theorem}[Polynomial space exact quantum algorithm for MST]\label{the:polyspace_q_mst}
    Given a graph $G = {(V, E)},\, K \subseteq V$ in the weighted adjacency matrix representation and let \(k=|K|,\, n = |V|,\, m=|E|\). There exists a quantum algorithm that computes the weight of the minimum Steiner tree $W_G(K)$ in time $O^*(2^{k})$ using polynomial bits and $O(poly(n,m,k))$ qubits.
\end{theorem}

The proof for theorem \ref{the:polyspace_q_mst} is in appendix \ref{sec:proof_mst_poly_space_q_exact}
.
\begin{theorem}[Quantum spacetime trade-off for MST]\label{the:q_spacetime_tradeoff_mst}
    Given a graph $G = {(V, E)},\, K \subseteq V$ in the weighted adjacency matrix representation and let \(k=|K|,\, n = |V|,\, m=|E|\). Let $l$ be the number of Dreyfus-Wagner recursion levels. There exists a quantum algorithm that computes the weight of the minimum Steiner tree $W_G(K)$ in time $O^*(e^{(1+g(k,l))k})$ using polynomial bits and $O(poly(n,m,k))$ qubits. Note that \(\lim_{k \to \infty, l \to \lfloor \log k\rfloor} g(k, l) = 0\).
\end{theorem}

\begin{proof}[Proof \ref{the:q_spacetime_tradeoff_mst}]\label{proof:qmst_poly_space_trade_off}

We begin by specifying the classical pre-computation and encoding complexities. We use theorem 2 from \cite{doi:10.1137/17M1140030} which guarantees the pre-computation steps are poly-space computable.

%\begin{theorem*}[Polynomial space exact algorithm for MST]\label{the:poly_space_exp_time_fomin_19}
%    For any \(\epsilon  > 0\) there is an \(n^{O (f' (\epsilon ))} logW \) space \(4^{(1+\epsilon )k}n^{O (f(\epsilon ))} logW \) time algorithm for Steiner Tree, where \(f\) and \(f'\)  are computable functions depending only on \(\epsilon\).
%\end{theorem*}
\begin{theorem*}[Polynomial space exact algorithm for MST]\label{the:poly_space_exp_time_fomin_19}
    There exists a polynomial-space algorithm for Steiner Tree running
in \(O(6.751^{k}n^{O (1)} logW) \) time.
\end{theorem*}

For a DW recursion of \(l_0\) levels with \(2^{l_0}\) branches, the classical pre-computation required to compute the sub-Steiner trees under a certain terminal set cardinality is identical across different branches on the \(l_0\)-th level and scales as the number of qualifying subsets times the time complexity for each. We keep an arbitrary constant \(C\) which can be substituted using different classical algorithms including \cite{doi:10.1137/17M1140030}. Let the maximum number of DW recursion layers be \(l_m = \lfloor{\log_2{k}}\rfloor\), we have the following.

\begin{equation*}\label{precomputation}
    T_{cla}(k,l_0)=\binom{k}{\frac{k}{2^{l_0}}} \times 2^{\log C \cdot \frac{k}{2^{l_0}}}
\end{equation*}

The encoding complexity counted in terms of 2-qubit gate count is given by the product of the number of qualifying subsets with the number of subgraphs in each subset.

\begin{equation*}\label{encoding}
    T_{enc}(k,l_0)=\binom{k}{\frac{k}{2^{l_0}}} \times {2^{l_0}}
\end{equation*}

The rounds of amplitude amplification follow the square root of the support size of the candidate solutions' subspace as a result of the quantum minimum finding process.

\begin{equation*}\label{q_search}
    T_{aa}(k,l_0)=\sqrt{\prod_{l=l_{m}-l_0+1}^{l=l_{m}} \binom{\frac{k}{2^{l-1}}}{\frac{k}{2^l}}}
\end{equation*}

Hence the overall complexity is the product of the sum of pre-computation and encoding complexities with the number of amplitude amplification rounds. We note that the pre-computation has to be repeated along with each round of encoding since polynomial space computations forbid the storage of exponentially large intermediate results. We use the following lemma.

\begin{lemma}[Binomial coefficient bound]\label{lem:binomial_coeff}
    \begin{equation*}
    \sqrt{\frac{n}{8k(n-k)}}  \exp{(n   H(\frac{k}{n}))} \leq {n \choose k} \leq \sqrt{\frac{n}{2 \pi k(n-k)}}  \exp{(n   H(\frac{k}{n}))}
    \end{equation*}
\end{lemma}

The overall complexity is given as below.

\begin{align*}\label{true_overall_complexity}
    T(k, l_0)
    &=(T_{cla}(k, l_0)+T_{enc}(k, l_0))\times T_{aa}(k, l_0)\\
    &= \left( \binom{k}{\frac{k}{2^{l_0}}} \times 2^{\log C \cdot \frac{k}{2^{l_0}}} + \binom{k}{\frac{k}{2^{l_0}}} \times {2^{l_0}} \right) \times \sqrt{\prod_{l=l_{m}-l_0+1}^{l={l_m}} \binom{\frac{k}{2^{l-1}}}{\frac{k}{2^l}}}  \\
    &\leq  \left(  \exp{(kH(\frac{\frac{k}{2^{l_0}}}{k}))}  \times ({2^{\log C \cdot \frac{k}{2^{l_0}}} + {2^{l_0}}})   \right) \times \sqrt{\prod_{l=l_{m}-l_0+1}^{l={l_m}}  \exp{\left( \frac{k}{2^{l-1}} H(\frac{\frac{k}{2^l}}{\frac{k}{2^{l-1}}}) \right)} } \qquad \text{[\ref{lem:binomial_coeff}]}  \\
    &=\left(  \exp{(kH(\frac{1}{2^{l_0}}))}  \times ({2^{\log C \cdot \frac{k}{2^{l_0}}} + {2^{l_0}}})   \right) \times \sqrt{\prod_{l=l_{m}-l_0+1}^{l={l_m}}  \exp{\left( \frac{k}{2^{l-1}} \right)}} \\
    &\leq \left(  \exp{(k\frac{\log {\frac{2^{l_0}}{e}}}{2^{l_0}})}  \times ({2^{\log C \cdot \frac{k}{2^{l_0}}} + {2^{l_0}}})   \right) \times   \exp{\left(\sum_{l=l_{m}-l_0+1}^{l={l_m}} \frac{k}{2^{l}} \right)} \qquad \text{[\( H(p) \leq p \log \frac{e}{ p} \)]}\\
    &=\left( \exp (k \cdot \frac{(l_0- \ln 2+2^{l_0}-1)}{2^{l_0}} ) \times (e^{\ln C \frac{k}{2^{l_0}}} + 2^{l_0}) \right)\\
    &=O^*(\exp (k (1+ \frac{l_0 + \ln C}{2^{l_0}}) ) )\\
\end{align*}

Therefore, \(g(k,l)=\frac{l+\ln C}{2^l}\) for \( 1\leq l \leq \lfloor \log k\rfloor, l\in \mathbb{N}\).

\end{proof}

The proof \ref{sec:proof_mst_poly_space_q_exact} for \ref{the:polyspace_q_mst} follows the above proof for \ref{the:q_spacetime_tradeoff_mst} with a simpler and tighter upper bounding enabled when substituting \(l_0=l_m\). We note that for the more realistic case of a likelihood based phylogenetic inference problem with real number weighted edges, this theorem \ref{the:polyspace_q_mst} has the potential to provide a speedup in the circuit model.
\section{Conclusions and Future Directions}

We present a quantum algorithm for the BNPP problem with polynomially better complexity bounds for the exponential parts in the best classical algorithm in the QRAM model. We also present a polynomial space quantum algorithm for MST in the circuit model. A future direction to obtain a QRAM-free quantum speedup in the setting of NPP is in approximating the weighted parsimony version of NPP, potentially building on \cite{10.1109/TCBB.2008.13} and \cite{doi:10.1137/1.9781611977073.128}. The general edge weight case warrants future investigation as the polynomial space QRAM-free quantum algorithm outperforms all classical counterparts.
% Another direction is to apply \cite{jeffery_quantum_2022} to the problem of unbalanced DW branch sizes. 
The main contribution of the algorithms is applying Grover's search algorithm to the divide-and-conquer procedure and using efficient Dicke state preparation to span the space of subgraphs. The complexity improvements will allow for calculating NPP for larger input sets in the presence of a fully fault-tolerant quantum computer. %Further methods dealing with recursive Grover must be studied to develop more efficient optimization algorithms based on classical results. The quantum minimum finding subroutine does not allow composing since%The fact that all current quantum dynamic programming algorithms fail to recover quadratic speedup might indicate certain fundamental structure, e.g., record keeping for previous computation, poses a hindrance to full quadratic speedup.
%Divide-and-conquer, which is essentially dynamical programming without the record keeping and 

\section{Acknowledgements}
A.H.W. thanks the VILLUM FONDEN for its support with a Villum Young Investigator (Plus) Grant  (Grant No. 25452 and Grant No. 60842) as well as via the QMATH Centre of Excellence (Grant No. 10059). A.H.W. acknowledges support
from the Novo Nordisk Foundation (GrantNo. NNF20OC0059939 ‘Quantum for Life’). S.B. acknowledges support from the Danish National Research Foundation via a chair grant (DNRF160) and the Novo Nordisk Foundation via The Novo Nordisk Young Investigator Award (NNF20OC0059309). S.B. acknowledges support from the Eric and Wendy Schmidt Fund For Strategic Innovation via the Schmidt Polymath Award (G-22-63345) which also supports L.M. In addition, L.M. thanks Dr. Joel Wertheim and Dr. Marc A. Suchard for their helpful discussions.

\printbibliography

\appendix

\section{Asymptotic growth of $q(n,m)$}\label{sec:q_asymptotic_growth}
We present a toy model below to demonstrate the growth of $q$ as a linear function of $mn$.

Several population genetics model exist, including the Wright-Fisher model \cite{Wright1955-gc} and the coalescent model \cite{Kingman1982-se}. A fundamental assumption underpinning these models is an infinite number of sites, an assumption that is difficult in modeling homoplasy events. Here, we do not employ any population genetics models. We will consider a binary phylogenetic tree, where the taxa have m binary independent characters and a Jukes-Cantor substitution model \cite{JUKES196921}. %This model does not reproduce the experimentally observed dependence of expected number of homoplasm events on the size of the input taxa size. This model produce an exponential dependence of $q(n)$.

%The definition of the input taxa size would need to be further using prior on the input taxa data. By excluding all back-propagation events, we optain a better.

For a binary tree with $d$ levels of offspring from a single root, the total number of offspring is $s=2^{d+1}-2$. Let $n=2^d$ be the number of offspring at depth $d$, which implies $s=2n-2$. Consider a simplistic mutation rule where the child state space and parent state space are related by a binary symmetric channel of crossover probability $p\ll1$. We count the number of m-homoplasm events without back mutations.

Consider $A^m_n$, the number of ways of choosing m non-empty compartments, each consecutive within itself, from a list of n indistinguishable elements.

It is easy to see the following recursion formula. The first term on the right is simple. Every term in the second summation on the right represents the scenario where we select \(n-r\) elements to the right as the rightest compartment and pick \(m-1\) compartments from the remaining elements.

$$A^m_n = A^m_{n-1} + \sum^{r=n-1}_{r = m-1}A^{m-1}_r $$

One can compute the expected number of homoplasm events on a single site as follows.

$$q_{ave}{(n,p)}=\mathbb{E}_{\textnormal{mutation rate} =  p}[q] = \sum_{m=1}^{m=n}(1-p)^{s-m}(p)^{m}A^{m}_{n}$$

Notice that $p\ll 1$. We can estimate $q_{ave}{(n,p)}$ via the leading terms in $\sum_{m=1}^{m=n}(p)^{m}A^{m}_{n}$.

By the above analysis, we see that the number of homoplasm events is proportional to the number of nodes of the phylogenetic tree leading to the observed set of taxa. The set of evolutionary histories we are looking at is the full set of possible evolutionary history (essentially binomial distribution) excluding the set of histories with back mutation events.

\section{Dreyfus-Wagner Recursion}\label{sec:dw_recursion}
We recall the DW recursion and explicitly lay out the 3 levels of DW recursion from \Romannum{1} to \Romannum{3}. Subscripts denote level number. Superscripts denote the order of appearance of definitions within the same subscript level and the same type of variable. We rename the variables from a previous presentation \cite{10.1007/978-3-030-58150-3_19} to make clearer the naming and avoid cross-level and cross-branch confusions. In the weight state preparation step, the tree is constructed bottom up and two branches that share the same variable here are meant to be tested for the relevant equality. The previous definition assumes that variable definition is only inherited by child formulas (formulas that contain the definition in question) which breaks down when considering different branches simultaneously. We need to consider different branches simultaneously when spanning the search space to obtain a quantum speed-up.
\begin{equation*}\label{def:dw}
W_G=
\min _ {\substack{ K_1 \subseteq  K \\ |K_1| =  (\alpha \pm \eta) }}
\min _ {\substack{ A \subseteq  V \\ |A| \leq  \lceil \log 1/\epsilon \rceil }}
W_{G}(K_1 \cup A) + W_{G/A}(K_2 \cup \{v_A\}) 
\end{equation*}

level \Romannum{1}

\begin{equation*}\label{def:l1}
W_G (K)=
\min _ {\substack{ K_1^1 \subseteq  K \\ |K_1^1| =  (\frac{1}{2} \pm \epsilon)k }}
\min _ {\substack{ A_1^1 \subseteq  V \\ |A_1^1| \leq  \lceil \log 1/\epsilon \rceil }}
W_{G}(K_1^1 \cup A_1^1) + W_{G/{A_1^1}}(K_1^2 \cup \{v_{A_1^1}\})
\end{equation*}
\[K_1^2 = K \setminus (K_1^1 \cup A_1^1)\]

level \Romannum{2}

\begin{equation*}\label{def:l2_b1}
W_G(K_1^1 \cup A_1^1)=
\min _ {\substack{ K_2^1 \subseteq  {K_1^1 \cup A_1^1} \\ |K_2^1| =  (\frac{1}{4} \pm O(\epsilon))k }}
\min _ {\substack{ A_2^1 \subseteq  V \\ |A_2^1| \leq  \lceil \log 1/\epsilon \rceil }}
W_{G}(K_2^1 \cup A_2^1) + W_{G/{A_2^1}}(K_2^2 \cup \{v_{A_2^1}\}) 
\end{equation*}

\[K_2^2 = ({K_1^1 \cup A_1^1}) \setminus ({K_2^1 \cup A_2^1})\]

\begin{equation*}\label{def:l2_b2}
W_{G/{A_1^1}}(K_1^2 \cup \{v_{A_1^1}\})=
\min _ {\substack{ K_2^3 \subseteq  {K_1^2 \cup \{v_{A_1^1}\}} \\ |K_2^3| =  (\frac{1}{4} \pm O(\epsilon)) }}
\min _ {\substack{ A_2^2 \subseteq  V \\ |A_2^2| \leq  \lceil \log 1/\epsilon \rceil }}
W_{G/{A_1^1}}(K_2^3 \cup A_2^2) + W_{(G/{A_1^1)}/{A_2^2}}(K_2^4 \cup \{v_{A_2^2}\}) 
\end{equation*}

\[K_2^4 = ({K_1^2 \cup \{v_{A_1^1}\}}) \setminus ({K_2^3 \cup A_2^2})\]

level \Romannum{3}

\begin{equation*}\label{def:l3_b1}
W_{G}(K_2^1 \cup A_2^1)=
\min _ {\substack{ K_3^1 \subseteq  {K_2^1 \cup A_2^1} \\ |K_3^1| =  (\frac{1}{4} \pm O(\epsilon))k }}
\min _ {\substack{ A_3^1 \subseteq  V \\ |A_3^1| \leq  \lceil \log 1/\epsilon \rceil }}
W_{G}(K_3^1 \cup A_3^1) + W_{G/{A_3^1}}(K_3^2 \cup \{v_{A_3^1}\}) 
\end{equation*}

\[K_3^2 = ({K_2^1 \cup A_2^1}) \setminus ({K_3^1 \cup A_3^1})\]

\begin{equation*}\label{def:l3_b2}
W_{G/{A_2^1}}(K_2^2 \cup \{v_{A_2^1}\})=
\min _ {\substack{ K_3^3 \subseteq  {K_2^2 \cup \{v_{A_2^1}\}} \\ |K_3^3| =  (\frac{1}{4} \pm O(\epsilon))k }}
\min _ {\substack{ A_3^2 \subseteq  V \\ |A_3^2| \leq  \lceil \log 1/\epsilon \rceil }}
W_{G/{A_2^1}}(K_3^3 \cup A_2^1) + W_{(G/{A_3^2})/{A_3^2}}(K_3^4 \cup \{v_{A_3^2}\})
\end{equation*}

\[K_3^4 = ({K_2^2 \cup \{v_{A_2^1}\}}) \setminus ({K_3^3 \cup A_3^2})\]

\begin{equation*}\label{def:l3_b3}
W_{G/{A_1^1}}(K_2^3 \cup A_2^2)=
\min _ {\substack{ K_3^5 \subseteq  {K_2^3 \cup A_2^2} \\ |K_3^5| =  (\frac{1}{4} \pm O(\epsilon))k }}
\min _ {\substack{ A_3^3 \subseteq  V \\ |A_3^3| \leq  \lceil \log 1/\epsilon \rceil }}
W_{G/{A_1^1}}(K_3^5 \cup A_3^3) + W_{(G/{A_1^1})/{A_3^3}}(K_3^6 \cup \{v_{A_3^3}\}) 
\end{equation*}

\[K_3^6 = ({K_2^3 \cup A_2^2}) \setminus ({K_3^5 \cup A_3^3})\]

\begin{equation*}\label{def:l3_b4}
W_{(G/{A_1^1})/{A_2^2}}(K_2^4 \cup \{v_{A_2^1}\})=
\min _ {\substack{ K_3^7 \subseteq  {K_2^4 \cup \{v_{A_2^1}\}} \\ |K_3^7| =  (\frac{1}{4} \pm O(\epsilon))k }}
\min _ {\substack{ A_3^4 \subseteq  V \\ |A_3^4| \leq  \lceil \log 1/\epsilon \rceil }}
W_{(G/{A_1^1})/{A_2^2}}(K_3^7 \cup A_3^4) + W_{({(G/{A_1^1})/{A_2^2}})/{A_3^4}}(K_3^8 \cup \{v_{A_3^4}\}) 
\end{equation*}

\[K_3^8 = ({K_2^4 \cup \{v_{A_2^1}\}}) \setminus ({K_3^7 \cup A_3^4})\]

\section{Proof of plugin lemma}\label{sec:plug_in}
In this section we present the full proof for the plugin lemma \ref{lem:mst_plugin_bnpp} for the BNPP problem using a MST subroutine. The framework is adopted from lemma 4.10 in \cite{sridhar_algorithms_2007}.

We will first prove the following lemma \ref{lem:steiner}.

\begin{lemma}[modified lemma 4.7 of \cite{sridhar_algorithms_2007}]\label{lem:steiner}
For a set of taxa M, if the number of nonisolated vertices of the associated conflict graph is t, then an optimum phylogeny $T^*_M$ can be constructed in time $O(C^{(s+t)} 2^t + nm^2)$, where \(s=penalty(M)\).
\end{lemma}

\begin{proof}
    We first construct the conflict graph and identify the nontrivial connected components of itin time \(O(nm^2)\). Let \(\kappa_i\) be the set of characters associated with component \(i\). We compute the MSt \(T_i\) for character set \(\kappa_i\). The remaining conflict-free characters in \(C \setminus \cup_i \kappa_i\) can be added by contracting each \(T_i\) to vertices and solving the perfect phylogeny problem using Gusfield's linear time algorithm.

    Since \(penalty(M)=s\), there are at most \(s+t=1\) distinct bit strings defined over character set \(\cup_i \kappa_i\). The Steiner space is bounded by \(2^t\) since \(|\cup_i \kappa_i|=t\). Using a MST subroutine running in time \(O(C^k)\), the total runtime for solving all Steiner tree instances is \( O(C^{(s+t)}2^t+nm^2)\)
\end{proof}

We can now give the full proof to lemma \ref{lem:mst_plugin_bnpp}.
\begin{proof}
    
We can declare each character to be in either a "marked" or unmarked state. At the beginning of the algorithm \ref{alg:classical_bnpp}, all of the characters are "unmarked". As the algorithm proceeds, we will mark characters to indicate that the algorithm has identified them as mutating more than once in the reconstracted tree \(T^*\).

We will examine two parameters, \(\rho \) and \(\gamma\), which specify the progress made by the derandomized algorithm in either identifying mulitple-mutating characters or reducing the problem to subproblems of lower total penalty. Consider the set of characters \(S\) such that, for all \(c \in S\), character \(c\)
is unmarked and there exists matrix \(M_i\) such that \(c\) mutates more than once in \(T^*_{M_i}\), the tree constructed from \(M_i\). We let \(\rho=|S|\). Parameter \(\rho\) refers to the number of characters mutating more than once that have not yet been identified. Parameter \(\gamma\) denotes the sum of the penalties of the remaining matrices \(M_i\), \(\gamma=\sum_i penalty(M_i)\).

Consider step 2a in algorithm \ref{alg:classical_bnpp}, when the algorithm selects character \(c(v)\). After selecting \(c(v)\), the algorithm proceeds to explore both cases when \(c(v)\) either mutates once or multiple times in \(T^*_{M_j}\). In the first case, \(penalty(T^*_{M_j})\) decreases by at least 1. Therefore \(\gamma\) decreases by at least 1. In the second case, the algorithm identifies a multiple mutant. We now mark character \(c(v)\), which reduces \(\rho\) by 1.
If the main loop at Step 3 terminates, then the algorithm finds optimal Steiner trees through a MST subroutine with base C and the runtime is bounded by \((2C^2)^\gamma\)  using lemma \ref{lem:steiner}. Let \(\alpha=(2C^2)^\gamma\). Therefore, the runtime of this portion of our algorithm can be expressed as the equation below.

\begin{align}
    T(\gamma,\rho) &\leq \max\{\alpha^\gamma, T(\gamma-1,\rho)+T(\gamma, \rho-1)+1\}
\end{align}

This quantity can be bounded by \(\alpha^{\gamma+1}(\frac{\alpha+1}{\alpha-1})^{\rho+1}\) which we can verify by induction.

\begin{align*}
T(\gamma,\rho) &\leq \max\{\alpha^\gamma, \alpha^{\gamma}(\frac{\alpha+1}{\alpha-1})^{\rho+1}+\alpha^{\gamma+1}(\frac{\alpha+1}{\alpha-1})^{\rho}+1\}\\
&=\max\{\alpha^\gamma,\alpha^\gamma (\frac{\alpha+1}{\alpha-1})^\rho(\frac{\alpha+1}{\alpha-1}+\alpha)+1\}\\
&=\max\{\alpha^\gamma,\alpha^\gamma (\frac{\alpha+1}{\alpha-1})^\rho(\frac{\alpha+1}{\alpha-1}+\alpha)+1\}\\
&\leq \max\{\alpha^\gamma,\alpha^\gamma (\frac{\alpha+1}{\alpha-1})^\rho(\frac{\alpha+1}{\alpha-1}+\alpha+1)\}\\
&=\alpha^{\gamma+1}(\frac{\alpha+1}{\alpha-1})^{\rho+1}
\end{align*}

Since we know that \(\gamma \leq q\) and \(\rho\leq q\), we can bound \(T(q,q)=O(\alpha^{q}(\frac{\alpha+1}{\alpha-1})^{q})=O^*((2C^2\frac{2C^2+1}{2C^2-1})^q)\). This yields an overall complexity of \(O^*((2C^2\frac{2C^2+1}{2C^2-1})^q+8^qnm^2)\) where the second term is the non-leading term from other analysis in \cite{sridhar_algorithms_2007}.

\end{proof}

\section{Proof for poly-space quantum algorithm}\label{sec:proof_mst_poly_space_q_exact}
We present the proof for theorem \ref{the:polyspace_q_mst} below.

\begin{proof} [Proof \ref{the:polyspace_q_mst}]\label{proof:qmst_poly_space}
    Recall the overall complexity from the proof for theorem \ref{the:q_spacetime_tradeoff_mst}. The only difference is that we let \(l_0=l_m\).
    \begin{align*}\label{true_overall_complexity}
    T(k, l_0)
    &=(T_{cla}(k, l_0)+T_{enc}(k, l_0))\times T_{aa}(k, l_0)\\
    &= \left( \binom{k}{\frac{k}{2^{l_0}}} \times 2^{\log C \cdot \frac{k}{2^{l_0}}} + \binom{k}{\frac{k}{2^{l_0}}} \times {2^{l_0}} \right) \times \sqrt{\prod_{l=l_{m}-l_0+1}^{l={l_m}} \binom{\frac{k}{2^{l-1}}}{\frac{k}{2^l}}}  \\
    &= \left( \binom{k}{\frac{k}{2^{l_m}}} \times 2^{\log C \cdot \frac{k}{2^{l_m}}} + \binom{k}{\frac{k}{2^{l_m}}} \times {2^{l_m}} \right) \times \sqrt{\prod_{l=1}^{l={l_m}} \binom{\frac{k}{2^{l-1}}}{\frac{k}{2^l}}} \qquad \text{[\(l_0=l_m=\lfloor \log k\rfloor\)]} \\
    &= \left( \binom{k}{1} \times 2^{\log C } + \binom{k}{1} \times {k} \right) \times \sqrt{\prod_{l=1}^{l={l_m}} 2^{\frac{k}{2^{l-1}}-1}} \qquad \text{[w.l.o.g., \(\lfloor \log k\rfloor=\log k\)]} \\
    &= \left( Ck + k^2 \right) \times \sqrt{\prod_{l=1}^{l={l_m}} 2^{\frac{k}{2^{l-1}}-1}} \\
    &= \left( Ck + k^2 \right) \times  2^{(k\sum_{l=1}^{l={l_m}}\frac{1}{2^{l}})-(\frac{l_m}{2})} \\
    &= \left( Ck + k^2 \right) \times  2^{k -1} \times \frac{1}{k^{\frac{1}{2}}} \\
    &= O^*(2^k)
    \end{align*}   
\end{proof}

\end{document}